\documentclass[a4paper,USenglish,cleveref,autoref]{article}
\usepackage{fullpage}
\usepackage{times,epsfig,epsf,psfrag,array,amsmath,amssymb,verbatim}
\usepackage{url,comment,enumerate,graphicx,graphics,wrapfig}
\usepackage[usenames,dvipsnames,svgnames,table]{xcolor}
\usepackage[ruled,linesnumbered,vlined]{algorithm2e}
\usepackage{multirow,amsthm}
\usepackage{pdfpages,framed,tcolorbox}
\usepackage{tabularx}
\usepackage{pdflscape,latexsym,bm,array,caption,textcomp}

\usepackage[square,sort,comma,numbers]{natbib}
\usepackage{url} 
\usepackage{hyperref}

\newtheorem{theorem}{Theorem}[section]

\newtheorem{lemma}[theorem]{Lemma}
\newtheorem{claim}[theorem]{Claim}
\newtheorem{proposition}[theorem]{Proposition}
\newtheorem{observation}[theorem]{Observation}

\newtheorem{definition}{Definition}

\renewcommand{\O}{{\cal O}}
\renewcommand{\P}{{\cal P}}
\newcommand{\C}{{\cal C}}

\newcommand{\freq}{\textsc{freq}}
\newcommand{\ftrs}{\textsc{FTRS}}
\newcommand{\kftrs}{\textsc{$k$-FTRS}}
\newcommand{\tO}{\widetilde O}

\title{New Extremal bounds for Reachability and Strong-Connectivity Preservers under failures}
\author{Diptarka Chakraborty\footnote{National University of Singapore, Singapore.  Supported in part by NUS ODPRT Grant, WBS No. R-252-000-A94-133. Email: diptarka@comp.nus.edu.sg} 
  \and Keerti Choudhary\footnote{Tel Aviv University, Israel. Email: keerti.choudhary@cs.tau.ac.il}}

\begin{document}
\pagenumbering{gobble}
\maketitle

\begin{abstract}
In this paper, we consider the question of computing sparse subgraphs for any input directed graph $G=(V,E)$ on $n$ vertices and $m$ edges, that preserves reachability and/or strong connectivity structures.

\begin{itemize}
\item We show $O(n+\min\{|\P|\sqrt{n},n\sqrt{|\P|}\})$ bound on a subgraph that is an $1$-fault-tolerant reachability preserver for a given vertex-pair set $\P\subseteq V\times V$, i.e., it preserves reachability between any pair of vertices in $\P$ under single edge (or vertex) failure. Our result is a significant improvement over the previous best $O(n |\P|)$ bound obtained as a corollary of single-source reachability preserver construction. 
We prove our upper bound by exploiting the special structure of single fault-tolerant reachability preserver for any pair, and then considering the interaction among such structures for different pairs.

\item In the lower bound side, we show that a $2$-fault-tolerant reachability preserver for a vertex-pair set $\P\subseteq V\times V$ of size $\Omega(n^\epsilon)$, for even any arbitrarily small $\epsilon$, requires at least $\Omega(n^{1+\epsilon/8})$ edges. This refutes the existence of linear-sized dual fault-tolerant preservers for reachability for any polynomial sized vertex-pair set. 

\item We also present the first sub-quadratic bound of at most $\tO(k~2^k~n^{2-1/k})$ size, for strong-connectivity preservers of directed graphs under $k$ failures. To the best of our knowledge no non-trivial bound for this problem was known before, for a general $k$. 
We get our result by adopting the color-coding technique of Alon, Yuster, and Zwick [JACM'95].
\end{itemize}
\end{abstract}

\newpage
\pagenumbering{arabic}
\section{Introduction}
One of the major problems in computer science, especially in the era of big data, is to \emph{sparsify} input graph while preserving certain properties of it. Let $\wp$ be any property defined over a graph. Given a graph $G=(V,E)$, a subgraph $H=(V,E_H)$, where $E_H \subseteq E$, is said to preserve property $\wp$ if the property $\wp$ is satisfied by the subgraph $H$ if and only if it is satisfied by graph $G$. Reachability and strong-connectivity are two fundamental graph properties that we consider in this paper. In case of reachability, given a directed graph $G$ and a set $\P$ of vertex-pairs the objective is to find a subgraph $H$ with as few edges as possible, so that for any pair $(s,t) \in \P$ there is a path from $s$ to $t$ in $H$ iff so is in $G$. This problem has been studied extensively~\cite{CE06, Bodwin17, AB18}. In case of strong-connectivity, given a directed graph $G$ the objective is to find a subgraph $H$ with as few edges as possible so that the strongly-connected components in $G$ and $H$ are identical. A folklore result shows that for any $n$-node graph we can have a strong-connectivity preserving subgraph with at most $2n$ edges.

In this paper we study the above two problems under the possibility of edge or vertex failures. 
In the real world networks are prone to failures. Most of the time such failures are unavoidable and also unpredictable in physical systems like communication or road networks. Due to this reason edge (or vertex) failure model draws a huge attention of the researchers in the recent past. In most of the scenarios such failures are much smaller in number in comparison to the size of the graph. Thus it is natural to associate a parameter to capture the number of edge (or vertex) failures, and then try to build fault-tolerant data-structures of size depending on this failure parameter for various graph theoretic problems. 
Many natural graph theoretic questions like connectivity~\cite{PP14, Parter15, BCR16, GK17}, finding shortest paths~\cite{DTCR08}, graph-structures preserving approximate distances~\cite{Luk99, CZ04, CLPR09, DK11, BK13, BGGLP15, BGPW17, BCHR18} etc. have been studied in the presence of edge (or vertex) failures. 

The main focus of this paper is to understand the extremal structure of paths in directed graphs under possible (bounded) edge failures. More specifically, our goal is to show existence (or non-existence) of subgraphs of certain size\footnote{Throughout this paper by \emph{size} of a subgraph we mean the number of edges present in that subgraph.} that preserves reachability and strong-connectivity in the presence of a small number of edge failures.

\begin{definition}[Fault-tolerant Strong-Connectivity Preserver (FT-SCC Preserver)]
For any graph $G=(V,E)$, a subgraph $H$ of $G$ is said to be $k$-fault-tolerant strong-connectivity preserver ($k$-FT-SCC preserver) if for each set $F~(\subseteq E)$ of $k$ edge failures, the strongly-connected components in $G-F$ and $H-F$ are identical. 
\end{definition}

We would like to emphasize that so far there is no non-trivial bound on the size of $k$-FT-SCC preserver. In this paper we show an upper bound of $\tO(k~2^k~n^{2-1/k})$ on $k$-FT-SCC preserver of any $n$-node (directed) graph. Moreover we show that we can find such a subgraph efficiently.

\begin{theorem}
There is a polynomial time (randomized) algorithm that given any directed graph $G = (V, E)$ on $n$ vertices and $k\geq 1$, computes a $k$-FT-SCC preserver of $G$ containing at most $\tO(k~2^k~n^{2-1/k})$ edges with probability at least $1-1/n^4$.
\end{theorem}
 
As a direct application we get a data-structure (oracle) of size sub-quadratic in $n$, for reporting strongly-connected components after $k$ failures.
The best earlier known bound for $k>1$ failures was $\Omega(2^k n^2)$~\cite{BCR19}, whereas for $k=1$ it was known by Georgiadis, Italiano, and Parotsidis~\cite{GIP17} that $O(n)$ space and query-time bound is achievable. 
So ours is the first truly sub-quadratic (i.e., $O(n^{2-\epsilon})$ for some $\epsilon>0$) sized strong-connectivity oracle for any constant number of failures. One may observe that there are graphs with $n$ nodes for which any $k$-FT-SCC preserver must be of size $\Omega(2^k n)$ (a simple proof of which is provided in Appendix~\ref{section:lower_bound}).

Next we study the extremal bounds of fault-tolerant reachability preserving subgraphs.

\begin{definition}[Fault-tolerant Pairwise Reachability Preserver]
For any graph $G=(V,E)$ and a set $\P$ of vertex-pairs, a subgraph $H$ of $G$ is said to be $k$-fault-tolerant pairwise reachability preserver for $\P$, denoted by \ftrs($\P,G$), if for each $F~(\subseteq E)$ of $k$ edge failure, the reachability relations between pairs in $\P$ agree on $G-F$ and $H-F$. 
\end{definition}

Baswana, Choudhary and Roditty~\cite{BCR16} provided a polynomial time algorithm that given any $n$-node directed graph constructs an $O(2^k n)$-sized subgraph that preserves reachability from a fixed source vertex to all other vertices under $k$ edge failures. 
As a corollary, to preserve reachability between arbitrary $\P$ pairs, we get an $O(2^kn|\P|)$-sized $k$-fault-tolerant pairwise reachability preserver. Clearly the bound is extremely bad for large sized set $\P$. So far we do not know any better bound even for small values of $k$. 
On the other hand in standard static setting (i.e., when $k=0$) we know existence of $O\big(n+(n|\P|)^{2/3}~\big)$-sized pairwise reachability preserver~\cite{AB18}.

An important question is how much the size of a reachability preserver varies when we go from standard static (i.e., without any failure) setting to single failure setting, and then further from single failure to dual failure setting. It is also natural to ask the following question: What is the bound on the number of pairs in $\P$ so that it is possible to obtain linear sized single-fault-tolerant pairwise reachability preservers. Here we show that this is possible as long as $|\P|=O(\sqrt{n})$. Note that this is also the current best known limit for the standard static setting~\cite{AB18}. So one cannot hope to improve our bound without improving the bound for static setting. Below we state our upper bound result.

\begin{theorem}
For any directed graph $G=(V,E)$ with $n$ vertices, and a set $\P$ of vertex-pairs, there exists a single-fault pairwise reachability preserver \ftrs($\P,G$) having at most $O\big(n+\min(|\P|\sqrt{n},~n\sqrt {|\P|})~\big)$ edges. Furthermore, we can find such a subgraph 
in polynomial time.
\end{theorem}

Our construction of SCC preservers plays a significant role in obtaining such a sparse reachability preserver. One may wonder whether the above result can be generalized to multiple failures, at least for constantly many failures. Unfortunately, we observe a striking difference between single and multiple (even for two) failures scenario in the context of pairwise reachability preservers.

\begin{theorem}
For every $n,p$ with $p=O(n^{2/3})$, there is an infinite family of $n$-node directed graphs and pair-sets $\P$ of size $p$, for which a dual fault-tolerant pairwise reachability preserver requires at least $\Omega(n |\P|^{\frac{1}{8}})$ edges.
\end{theorem}

This shows linear size reachability preserver is not possible under dual failures in general pairwise setting with the number of pairs being $\Omega(n^{\epsilon})$ for any small $\epsilon > 0$. As a consequence we get a polynomial separation in size of a pairwise reachability preserver between single and dual failures. This is in sharp contrast with single-source all destinations setting, wherein, the size only doubles each time we increase the count of failure by value one~\cite{BCR16}.

It is worth mentioning that in this paper we show fault-tolerant structures with respect to edge failures only, however, all our results hold for vertex-failures as well.

\subsection{Related Work}
A simple version of reachability preserver is when there is a single source vertex $s$ and we would like to preserve reachability from $s$ to all other vertices. Baswana \emph{et al.}~\cite{BCR16} provided an efficient construction of a $k$-fault-tolerant single-source reachability preserver of size $O(2^kn)$. Further they showed that this upper bound on size of a preserver is tight up to some constant factor. As an immediate corollary of their result, we get a $k$-fault-tolerant pairwise reachability preserver of size $O(2^kn|\P|)$ (by applying the algorithm of~\cite{BCR16} to find subgraph for each source vertex in pairs of $\P$, and then taking the union of all these subgraphs). We do not know whether this bound is tight for general $k$. However for standard static (with no faulty edges) setting much better bound is known. We know that even to preserve all the pairwise distances, not just reachability, there is a subgraph of size $O\big(n+\min(n^{2/3}|\P|, n\sqrt{|\P|})~\big)$~\cite{CE06, Bodwin17}. 
Later Abboud and Bodwin~\cite{AB18} showed that for any directed graph $G=(V,E)$ given a set $S$ of source vertices and a pair-set $\P\subseteq S \times V$ we can construct a pairwise reachability preserver of size 
$O\big(n+\min(\sqrt{n |\P| |S|}, (n|\P|)^{2/3})~\big)$. 
It is further shown that for any integer $d \ge 2$ there is an infinite family of $n$-node graphs and vertex-pair sets $\P$ for which any pairwise reachability preserver must be of size $\Omega\big(n^{2/(d+1)}|\P|^{(d-1)/d}\big)$. Note, for undirected graph storing spanning forests is sufficient to preserve pairwise reachability information, and thus we can always get a linear size reachability preserver for undirected graphs. We would like to emphasize that all the results provided in this paper hold for directed graphs. A problem similar to constructing reachability preserver is to construct a data-structure (aka. oracle) that can answer queries of the form whether a vertex is reachable from a fixed source vertex $s$ after multiple edge (or vertex) failures. As an application of~\cite{BCR16} we get such an oracle of size $O(2^k n)$ for $k$ edge (or vertex) failures with query time $O(2^k n)$. For just dual failures we have an $O(n)$ size oracle with $O(1)$ query time due to~\cite{Choudhary16}.
In a recent work, Brand and Saranurak~\cite{BrandS19} obtained a $k$-fault-tolerant $\O(n^2)$ sized reachability oracle that has $O(k^\omega)$ query time, where $\omega$ is the constant of matrix-multiplication.

Finding strongly connected components (SCCs) under edge failures is another important problem. One specific problem is given a directed graph $G$ to build a data-structure (oracle) that for any vertices $u,v$ and a set of edges $F$ of size $k$ can answer whether $u$ and $v$ are in the same SCC in $G-F$. Using $k$-fault-tolerant reachability preserver of~\cite{BCR16} we can get such an oracle of size $O(2^k n^2)$ with query time $O(2^k n)$ (see~\cite{BCR19}). Moreover,~\cite{BCR19} provides us an algorithm that computes all the SCCs in $G-F$ in time $O(2^k n \log^2 n)$ by using a data-structure of size $O(2^k n^2)$. Georgiadis, Italiano and Parotsidis~\cite{GIP17} also studied this problem of computing all the SCCs under single edge failure, and gave a solution with $O(n)$ query time using a data-structure of size only $O(n)$. However so far we do not know any solution for computing all the SCCs after more than one edge failures using a data structure of size $O(n^{2-\epsilon})$ for any $\epsilon>0$. In this paper we give construction of first such truly sub-quadratic sized data-structure as long as there are only constantly many failures.
For undirected graphs, the optimal bound of $O(kn)$ edges for $k$-fault-tolerant connectivity preserver directly follows from $k$-edge (vertex) connectivity certificate constructions provided by Nagamochi and Ibaraki~\cite{NagamochiI:92}. In contrast, for directed graphs, the only known result for `sparse' certificate of $k$-edge (vertex) strong-connectivity is for $k=2$, due to a series of works by Georgiadis et al. \cite{GeorgiadisILP15-soda, GeorgiadisILP15-icalp, GeorgiadisIKPP17}.
Our truly sub-quadratic sized $k$-FT-SCC preservers also in turn provides the {\em first} truly sub-quadratic sized $k$-edge (vertex) strong-connectivity certificates for directed graphs, for $k\geq 3$ (see Section~\ref{section:applications}).

Other closely related problems that have been studied in the fault-tolerant model include computing distance preservers~\cite{DTCR08, PP13, Parter15}, depth-first-search tree~\cite{BCCK19}, spanners~\cite{CLPR09, DK11}, approximate distance preservers~\cite{BK13, PP14, BGLP16}, approximate distance oracles~\cite{DP09, CLPR10}, compact routing schemes~\cite{CLPR10, Chechik13}.

\subsection{Technical Overview}

\subparagraph{SCC preserver}
Our starting point is a simple construction, which is motivated from some of the techniques used in~\cite{GIP17}, of linear (in number of vertices) sized single fault-tolerant SCC oracle. Then by using that construction as a basic building block we provide a construction of fault-tolerant SCC preserver for $k$ edge failures. Using the ideas inspired by color-coding technique of Alon, Yuster and Zwick~\cite{AlonYZ95}, we show a generic procedure that converts any $r$-fault-tolerant (or even non-fault-tolerant) SCC preserver construction into a $(k+r)$-FT-SCC preserver construction. Our technique, especially the first step of our conversion procedure, is quite similar to that used in~\cite{DK11} to convert any spanner to a $r$-fault-tolerant spanner with the same \emph{stretch}. The first step alone cannot serve our purpose fully, mostly because it could work (with high probability) only when the SCCs after $k$ faults are of "small size". To mitigate this issue we have to handle the large sized SCCs (after $k$ failures) with a completely different technique. As a consequence our whole proof becomes slightly more intricate than that in~\cite{DK11}. 

Our conversion procedure works as follows. In the first step, we sample a set $J$ of edges and treat $J$ as failure set, and then compute 
$r$-fault-tolerant SCC preserver
of $G-J$ (residual graph after removing edges in $J$). Do this multiple times and take union of all those
$r$-FT SCC preservers for different random choices of $J$. Now in the second step, we sample a set $W$ of "a few" vertices of $G$, and then for each $w\in W$ compute single-source {\ftrs} with $w$ as the source and single-destination {\ftrs} by treating $w$ as the destination. Then take union of all these {\ftrs} subgraphs. Finally we claim that with high probability the union of subgraphs produced by first and the second step is a $(k+r)$-FT-SCC preserver.

Our correctness proof proceeds as a win-win analysis. For the sake of simplicity let us provide a high level proof sketch for the case when $r=0$. For any set $F$ of edge failures, we distinguish two cases depending on whether a SCC $C$ in $G-F$ is small or not. Our choice on size of $J$ ensures that we over-samples $F$ during the first step. Hence if $C$ is of small size with high probability at least for one random choice of $J$, $C$ will be completely inside one strongly connected component after removing $J$ (that also includes $F$) from $G$ (in other words, $J$ "separates" $C$ from $F$), leading to $C$ also being a SCC in the final subgraph (after failure of edges in $F$). Next we turn to the case when $C$ is of large size. In that scenario it is not difficult to show that with high probability $W$ (chosen at random during second step) and $C$ have some common vertex, and hence $C$ will be preserved due to inclusion of single-source and single-destination {\ftrs} structure. Our techniques hold even when we are able to "partially separate" $C$ from $F$, and that helps us in proving our result for any $r$ (see Section~\ref{sec:SCC preserver} for the details). So we get that any improvement in size of $r$-FT-SCC preserver will directly improve the size of $(k+r)$-FT-SCC preserver.
%

\vspace{-2mm}
\subparagraph{FT-Reachability-Preserver}
The construction of $O(n+\sqrt{n}|\P|)$-sized reachability-preserver for a pair-set $\P$ uses the fact the preservers on general digraphs are reducible to preserver on DAGs (since the SCCs can always be compressed into "supernodes", and there is a linear size certificate for strong-connectivity). In a DAG, it is not very difficult to ensure that paths between two given pairs meet and diverge only once, which in turn provides a cap on the maximum number of edges in a preserver. Our approach to FT-reachability is to try to adapt the constructions for non-faulty setting~\cite{AB18}. However, one major hindrance we face is that we cannot directly compress SCCs into "supernodes" as they can destroy $2$-connectivity structures, and thus the problem is not reducible to DAGs. We start by observing that \ftrs$(p)$ for a pair $p=(s,t)$ is just union of two "maximally disjoint" $s-t$ paths. The interactions between \ftrs$(p)$ for different pairs $p$ help us in achieving our bound of $O(n+\sqrt{n}|\P|)$. Our second upper bound of $O(n \sqrt{|\P|})$ is much simpler than the first one. Again we consider union of \ftrs$(p)$ structures for all $p \in \P$, and then consider all the vertices that appear in more than $\sqrt{|\P|}$ "maximally disjoint" paths in total (in all {\ftrs} structures). Next we remove all those paths and add single-source and single-destination {\ftrs} structure from those selected vertices. Then we use the properties of single-source and single-destination {\ftrs} to show that the final subgraph will be a {\ftrs} for the pair-set $\P$.
%


These structures as well as linear bound on preserver size for small sized $\P$~(at most $O(\sqrt n)$ pairs), is not expendable beyond single failure due to the fact that $\kftrs(p)$ for any $k>1$ cannot be represented as union of $o(n)$ paths (see \cite{Choudhary16}). In the latter part of this paper we also show that this is not a drawback of our approach, instead in some sense it is unavoidable, by proving a size lower bound of any $k$-fault-tolerant pairwise preserver, for $k\geq 2$. 
Our lower bound for dual failure is inspired by the following observation: If for a pair $p=(s,t)$, $Q_p=(q_1,\ldots,q_\ell)$ and $R_p=(r_1,\ldots,r_t)$ are two vertex-disjoint paths from $s$ to $t$. Then by playing with failures on $Q$ and $R$, we can force multiple paths originating from $Q$ and terminating to $R$ to be present in our $2$-FTRS for $p$. Note that a $2$-FTRS, for a single pair $p$, would still be linear in size. However, as the number of pairs increases achieving sparsity is tricky. 
To obtain a $2$-FTRS lower-bound for multiple pairs we embed in between the paths $Q_p$ and $R_p$, the "hard" distance preserver graph given by Coppersmith and Elkin~\cite{CE06}. We start with a lower bound distance preserver graph $G$ over pair-set $\P$, and perform its layering $L$ number of times, for some parameter $L$. Inspired by techniques of Bodwin et al.~\cite{Bodwin17,AB18}, we are able to show that all the paths in the "hard" instance graph from~\cite{CE06} can be assumed to have equal distance between the relevant pairs. Thus our layered embedded structure also acts as a non-faulty reachability preserver among pairs with end-points respectively on 
paths $Q_p$ and $R_p$.



%
%

\section{Preliminaries and Tools}\label{section:prelims}

Given a directed graph $G=(V,E)$ on $n=|V|$ vertices and $m=|E|$ edges, 
the following notations will be used throughout the paper.\\[-2mm]
\begin{itemize}
\item $H[A]$~:~ The subgraph of $H$ induced by vertices in set $A$.
\item $G^R$~:~ The graph obtained by reversing all the edges in graph $G$.
\item $H- F$~:~ For a set of edges $F$, the graph obtained by deleting the edges in $F$ from graph $H$.
\item $\pi(x,y,H)$~:~ The shortest path from $x$ to $y$ in graph $H$.
\item $P \circ Q$~:~ The concatenation of two paths $P$ and $Q$, i.e., a path that first follows $P$ and then $Q$.
\item $T(v)$~:~ The subtree of a directed tree $T$ rooted at a vertex $v\in T$.
\item $cert(C,H)$~:~ An arbitrarily chosen certificate of at most $2(|C|-1)$ edges corresponding to a strongly connected component $C$ in $H$.

\end{itemize}

Our algorithm for computing {\em pairwise-reachability} and {\em strong-connectivity} preservers in a fault tolerant environment employs the concept of a {\em single-source} {\ftrs} which is a sparse subgraph that preserves reachability 
from a designated source vertex even after the failure of at most $k$ edges in $G$. 
Observe that in case of no failure, a directed reachability tree has $n-1$ edges and is able to preserve reachability from the source which is a also the root. 
An {\ftrs} with respect to a given source is formally defined as follows.

\begin{definition}[$\ftrs$]
Let $\P\in V\times V$ be any set of pairs of vertices. A subgraph $H$ of $G$ is said to be a $k$-Fault-Tolerant Reachability-Subgraph of $G$ for $\P$ if for any pair $(s,t)\in \P$ and for any subset $F\subseteq E$ of $k$ edges, $t$ is reachable from $s$ in $G - F$ if and only if $t$ is reachable from $s$ in $H - F$. Such a subgraph $H$ is denoted by $\kftrs(\P,G)$, or simply $\ftrs(\P,G)$ when $k=1$.
\label{definition:FTRS}
\end{definition}

Baswana \emph{et al.}~\cite{BCR16} provide a construction of sparse $\ftrs$ for any general $k\ge 1$ when there is a designated source vertex.

\begin{theorem}[\cite{BCR16}]
\label{theorem:ftrs}
For any directed graph $G=(V,E)$, a designated source vertex $s\in V$, and an integer $k\geq 1$, there exists a (sparse) subgraph $H$ of $G$ which is a $\kftrs(\{s\}\times V,G)$ and contains at most $2^k n$ edges. Moreover, such a subgraph is computable in $O(2^kmn)$ time, where $n$ and $m$ are respectively the number of vertices and edges in graph $G$.
\end{theorem}

Our constructions will require the knowledge of the vertices reachable from a vertex $s$ as well as the  vertices that can reach $s$. So we will be using $\ftrs$ defined with respect to a source vertex ($\{s\}\times V$ case), as well as 
$\ftrs$ defined with respect to a destination vertex ($V\times \{s\}$ case). 

In this paper, we consider fault-tolerant structures with respect to edge failures only.
Vertex failures can be handled by simply splitting a vertex $v$ into an edge $(v_{in},v_{out})$,
where the incoming and outgoing edges of $v$ are respectively directed into $v_{in}$ and directed out
of $v_{out}$.

\section{Strong-connectivity Preservers}
\label{sec:SCC preserver}
We show a construction of strong-connectivity preservers that are able to preserve strong-connectivity relation between vertices in $G=(V,E)$ as long as the number of failures are bounded by $k$. For our convenience, we assume that $G$ is strongly connected, if not, we may apply our construction to each strongly-connected component (SCC) of $G$. Although the main contribution of this section is to get a fault-tolerant SCC preserver for general $k$ failures, let us start with the case when there can be at most one edge failure.

\subsection{Construction for single failure}
We first give a simple construction of an $O(n)$ size FT-SCC preserver for the scenario of $k=1$. Let $s$ be an arbitrary vertex in $G$. We initialize $H_1$ to union of subgraphs $\ftrs(\{s\}\times V,G)$ and $\ftrs(V\times \{s\},G)$. The following simple observation describes the significance of $H_1$ in preserving strong connectivity information.

\begin{observation}
\label{obs:SCC-H1}
Given a directed graph $G$, let $H_1$ be the union of subgraphs $\ftrs(\{s\}\times V,G)$ and $\ftrs(V\times \{s\},G)$. For any vertex $x$ and any edge-failure $e$, if $x$ and $s$ are strongly connected in $G- \{e\}$, then they are also strongly connected in $H_1- \{e\}$.
\end{observation}

We now introduce a lemma 
that will be crucial in preserving strong-connectivity between vertices not (strongly) connected to $s$ after failure.

\begin{lemma}
\label{lemma:list}
For any $n$-vertex directed graph $G=(V,E)$ and an ordered list $L=(v_1,v_2,\ldots,v_n)$ of vertices of $G$, in polynomial time we can compute a subgraph $H_0=H_0(L)$ of $G$ with at most $2n$ edges satisfying the condition that the SCCs of $G[v_1\cdots v_i]$ are identical to those in $H_0[v_1\cdots v_i]$, for $1\le i\le n$.
\end{lemma}
\begin{proof}
For any $i \in \{1,\cdots,n\}$, let $V_i=\{v_1,\ldots,v_i\}$ be the subset of $V$ comprising of first $i$ vertices, and $G_i$ be the subgraph of $G$ induced by the set $V_i$. We initialize $H_0$ to be an empty graph on $n$ vertices. The edges of $H_0$ are incrementally computed in $n$ rounds, wherein, in the $i^{th}$ round we add edges to $H_0$, so as to ensure that the SCCs of  $H_0[v_1\cdots v_i]$ are identical to those in $G_i$.

For any $i\in \{1,\cdots,n\}$, let $\gamma_i$ denote the number of SCCs in graph $G_i$, and let $C_i$ be the SCC of $v_i$ in graph $G_i$. 
Observe that $\gamma_{i}\leq 1+\gamma_{i-1}$, where the equality holds for index $i$ if and only if $C_i=\{v_i\}$. If $C_{i,1},C_{i,2},\ldots,C_{i,\ell_i}$ is a decomposition of SCC $C_i$ in $G_{i-1}(=G_i- v_i)$, then in round $i$ it suffices to add at most $2\ell_i$ edges corresponding to an out-reachability and an in-reachability tree rooted at $v_i$ and spanning the "super-nodes" (that is obtained by contracting all the edges in a component $C_{i,j}$) $C_{i,1},C_{i,2},\ldots,C_{i,\ell_i}$. 
Now $\ell_i=1+\gamma_{i-1}-\gamma_{i}$.
Thus the number of edges in $H_0$ is at most $2(\ell_2+\ldots+\ell_n)=2(n-1+\gamma_1-\gamma_n)\leq 2n$.
\end{proof}

\paragraph*{Construction procedure of $1$-FT-SCC preserver}
Consider an arbitrarily chosen vertex $s$ in $G$. Then by treating $s$ as a source vertex, we compute the directed reachability-tree $T$ rooted at $s$ for graph $G$. Similarly we compute a reachability-tree $T'$ for reverse graph $G^R$. Let $L$ (resp. $L'$) represent an ordered list containing the vertices of $T$ (resp. $T'$) sorted in the decreasing order of their depth (where vertices in the same depth are in an arbitrary order). Next we compute the subgraphs $H_0(L)$ and $H_0(L')$ with the property mentioned in Lemma~\ref{lemma:list}, and finally set $H$ to be the union of graphs $H_1$ (as defined in Observation~\ref{obs:SCC-H1}), $H_0(L)$, and $H_0(L')$. 

It is easy to see that $H$ contains $O(n)$ edges. We now prove the correctness.

Consider a vertex $x$ in $G$ and a failing edge $e=(a,b)$ such that $x$ and $s$ are not strongly-connected in $G- \{e\}$. Let $C_x$ be the SCC of $x$ in $G- \{e\}$. We will show that $C_x$ must be an SCC in at least one of the graphs: $H_0(L)- \{e\}$ or $H_0(L')- \{e\}$.

Observe that $x$ is either not reachable from $s$ in $G- \{e\}$, or does not have a path to $s$ in $G- \{e\}$. Without loss of generality, we assume that the first case holds. Thus $e=(a,b)$ must lie on the tree-path from $s$ to $x$ in $T$. Then $a$ is a parent of $b$ in $T$. 
Observe that since none of the vertices of $C_x$ can be reachable from $s$ in $G- \{e\}$, the entire SCC $C_x$ must lie in the subtree rooted at $b$, denoted by $T_b$. Since $L$ stores vertices of $T$ sorted in the decreasing order of depth, the vertices of subtree $T_b$ (and hence also $C_x$) appears before $a$ in the list $L$. This implies that $C_x$ must be an SCC in $H_0(L)- \{(a,b)\}$. This completes the correctness.


So we conclude with the following theorem.

%

\begin{theorem}~\label{theorem:SCC_1}
There is a polynomial time (deterministic) algorithm that given any directed graph $G=(V,E)$ on $n$ vertices, computes an $1$-FT-SCC preserver of $G$ with at most $O(n)$ edges.
\end{theorem}

\subsection{A generic construction}

In this section we provide a construction for general $k$ failures.

\begin{lemma}
\label{lemma:SCC_k}
If there is an algorithm $\mathcal{A}$ that on every $n$-node directed graph builds a $r$-fault-tolerant SCC preserver of size $f(n,r)$, then for any $k = \Omega(r)$, there is a randomized algorithm $\mathcal{B}$ that given a directed graph $G$ and a parameter $\alpha\in [0,1]$, computes a $(k+r)$-FT-SCC preserver of size $O(k2^{k+r}\cdot n^{2-\alpha} \log n + n^{k\alpha} \cdot \log n\cdot f(n,r))$ with high probability. Moreover, if $\mathcal{A}$ runs in time $T(n)$ then the algorithm $\mathcal{B}$ runs in time $poly(n)T(n)$.
\end{lemma}

\paragraph*{Description of Procedure $\mathcal{B}$ to compute $(k+r)$-FT-SCC preserver}
Let $\alpha\in [0,1]$ be the input parameter. Procedure $\mathcal{B}$ constructs two graphs $H_1$ and $H_2$ as follows.
\begin{itemize}
\item {\bf \emph H$_1$ :}
Repeat the following for $L=16\cdot n^{k\alpha} \cdot \log n$ number of iterations: Independently add each edge of $G$ to a set $J$ with probability $p = \frac{2}{n^{\alpha}}$, and then use the given algorithm $\mathcal{A}$ to compute a $r$-FT-SCC preserver of the remaining graph $G - J$. Set $H_1$ to be union of these $r$-FT-SCC preservers, taken over all $L$ iterations.
\item {\bf \emph H$_2$ :}
Let $q=16(k+r)\cdot n^{-\alpha} \log n$, and $W$ be a uniformly random set of $nq=\big(16(k+r)\cdot n^{1-\alpha} \log n\big)$ vertices in $G$. Initialize $H_2$ to be union of $(k+r)$-$\ftrs(\{w\}\times V,G)$ and $(k+r)$-$\ftrs(V\times \{w\},G)$, taken over all $w\in W$.
\end{itemize}
Finally procedure $\mathcal{B}$ outputs a new subgraph $H$ which is union of $H_1$ and $H_2$.


\paragraph*{Correctness of Procedure~$\mathcal{B}$}
 For a set $F$ of edge failures and a vertex $x$ in $G$, let $C_{x,F}$ be the SCC containing $x$ in $G- F$. We say that the SCC $C_{x,F}$ is small if it contains at most $\frac{n^{\alpha}}{4}$ vertices, and large otherwise.

First we will consider the scenario that $C_{x,F}$ is small. 
Let $F_1$ and $F_2$ be any two disjoint subsets of $E$ of size respectively $k$ and $r$, and let $F=F_1\cup F_2$.
(Observe that $C_{x,F_1}$ might be large even though $C_{x,F}$ is small).
Let $cert(C_{x,F},G- F)$ be an arbitrary certificate of at most $2(|C_{x,F}|-1)$ edges corresponding to a strongly connected component of $C_{x,F}$. We say an iteration \emph{separates} $C_{x,F}$ from $F_1$ if at that iteration none of the edges of $cert(C_{x,F},G- F)$ is selected in $J$, but all the edges of $F_1$ lie in $J$. The probability that a particular iteration separates $C_{x,F}$ from $F_1$ is:

$$(1-p)^{|cert(C_{x,F},G- F)|}\cdot p^{|F_1|}\geq \Big(1-\frac{2}{n^\alpha}\Big)^{2(n^\alpha/4)}\cdot \Big(\frac{2}{n^\alpha}\Big)^k\geq \frac{1}{4}\cdot \frac{2^k}{n^{k\alpha}}~.$$

The probability that none of the iterations is able to separate $C_{x,F}$ from $F_1$ is at most
$$\Big(1-\frac{2^k}{4 n^{k\alpha}}\Big)^{16n^{k\alpha} \cdot \log n}\leq \frac{1}{n^{4(2^k)}}~.$$

Now, there are $n^{O(k)}$ (assuming $k = \Omega(r)$) choices for pair $(F_1,F_2)$, and $n$ choices for $x$, thus a total of $n^{O(k)}$ different choices for the triplet $(C_{x,F},F_1,F_2)$. 
By union bound, the probability that none of the $L$ iterations are able to separate $C_{x,F}$ from $F_1$, for at least one choice of $(C_{x,F},F_1,F_2)$, is at most:
$\frac{n^{O(k)}}{n^{4(2^k)}}\leq \frac{1}{n^5}$.

The next claim is immediate from definition of FT-SCC preserves.


\begin{claim}~\label{claim:scc}
Let $F_1, F_2, J \subseteq E$ where $J$ contains $F_1$, $F$ be $F_1\cup F_2$, and $C$ be an SCC in $G-F$ whose certificate is disjoint with $J$. 
Further let $\widetilde H$ be a $r$-FT-SCC preserver of $G-J$. Then $S$ is also an SCC in $\widetilde{H} - F_2$
if $|F_2|\leq r$.
\end{claim}


The above discussion together with Claim~\ref{claim:scc} completes the analysis of the scenario when the SCCs are small, and we obtain the following  lemma.

\begin{lemma}
$H_1$ with high probability preserves small SCCs after $k+r$ failures.
\label{lemma:small_sccs}
\end{lemma}

Next, we consider the scenario that $C_{x,F}$ is large, i.e, contains more than $\frac{n^{\alpha}}{4}$ vertices. 
Let $F$ be a set of $k+r$ edge failures, and $C_{x,F}$ be the SCC of $x$ in $G- F$. In order to ensure that the SCC $C_{x,F}$ is preserved in $H_2- F$ it suffices to ensure that $C_{x,F}$ has non-empty intersection with $W$. This is because we include in $H_2$ the $(k+r)$-fault-tolerant in-and-out-reachability preserves of each of the vertices of   $W$ in $H_2$. The probability\footnote{It is easy to verify that the bound of $\frac{1}{n^{4(k+r)}}$ holds for both the scenarios: sampling with replacement by probability $q$, or just taking $W$ to be a uniformly random subset of vertices of $nq$ size.} 
that none of the vertices of $C_{x,F}$ lie in random set $W$ is at most:

$$(1-q)^{|C_{x,F}|}\leq \Big(1-\frac{16(k+r)\log n}{n^\alpha}\Big)^{n^\alpha/4}
\leq \frac{1}{n^{4(k+r)}}~.$$

Again, there are a total of $n^{2(k+r)+1}$ different choices for the pair $(C_{x,F},F)$. By union bound, the probability that for at least one choice of $(F,x)$, there is some large SCC $C_{x,F}$ having non-empty intersection with $W$ is at most: $\frac{n^{2(k+r)+1}}{n^{4(k+r)}}\leq \frac{1}{n^5}$.
This completes the analysis of the scenario when the SCCs are large.

\begin{lemma}
$H_2$ with high probability preserves large SCCs after $k+r$ failures.
\label{lemma:large_sccs}
\end{lemma}

Now we are ready to prove Lemma~\ref{lemma:SCC_k}.
\begin{proof}[Proof of Lemma~\ref{lemma:SCC_k}]
Recall, $H$ is the graph obtained by taking the union of the graphs $H_1$ and $H_2$. From Lemma~\ref{lemma:small_sccs} and Lemma~\ref{lemma:large_sccs}, it follows that our construction results in a valid $(k+r)$-fault-tolerant SCC preserver. The total number of edges in $H$ is $O(k2^{k+r}\cdot n^{2-\alpha} \log n + n^{k\alpha} \cdot \log n\cdot f(n,r))$.
\end{proof}

Now as a corollary of Lemma~\ref{lemma:SCC_k}, we directly obtain a $k$-FT-SCC preserver with sub-quadratic in $n$ edges, since we know $f(n,0)=O(n)$. However to get even better bound we use $f(n,1)=O(n)$ by Theorem~\ref{theorem:SCC_1}. On substituting $r=1$ and $\alpha=1/(k+1)$ in Lemma~\ref{lemma:SCC_k}, we obtain that a $(k+1)$-FT-SCC preserver has at most $\tO(k~2^k~n^{2-1/(k+1)})$ edges.
Thus the following theorem is immediate.

\begin{theorem}~\label{theorem:kFT-scc-preserver}
For every digraph $G = (V, E)$ on $n$ vertices and every $k\geq 1$, there is a polynomial time (randomized) algorithm that with probability at least $1-1/n^4$ computes a $k$-fault-tolerant SCC preserver of $G$ with at most $\tO(k~2^k~n^{2-1/k})$ edges.
\end{theorem}

\section{Reachability Preservers}
\label{sec:reachability preserver} 

In this section we will focus on finding a sparse pairwise reachability preserver. Recall that, for a directed graph $G=(V,E)$ and a set $\P$ of vertex-pairs, a $1$-fault tolerant reachability subgraph of $G$ is a subgraph $H$ that preserves the reachability information between all pairs of vertices in $\P$ under single edge failure. We denote such a subgraph by \ftrs($\P,G$), or simply \ftrs($\P$) if the underlying graph $G$ is clear from the context. 

\subsection{Upper Bound I}

Let us start by showing an existential upper-bound on the number of edges present in an optimum sized \ftrs. 

\begin{theorem}
\label{thm:existential-UB-1}
For any directed graph $G=(V,E)$ with $n$ vertices, $m$ edges, and a set $\P$ of vertex-pairs, there exists a \ftrs($\P,G$) (or simply \ftrs($\P$)) that contains $O(n+|\P|\sqrt{n})$ edges.
\end{theorem}
As a corollary of the above theorem we get a \ftrs of linear (in number of vertices) size whenever number of pairs for which we have to preserve reachability is at most $O(\sqrt{n})$. Note, for each of $O(\sqrt{n})$ pairs if we use Theorem~\ref{theorem:ftrs} separately we get a \ftrs of size $O(n^{3/2})$. Hence our result improves the size of a \ftrs by a factor of $\sqrt{n}$. We devote this subsection to prove the above theorem.

Let $H_{scc}$ be an $1$-fault tolerant SCC preserver of $G$ as obtained by Theorem~\ref{theorem:SCC_1}, and $H_{opt}$ be an optimum sized subgraph of $G$ such that $H:=H_{scc}\cup H_{opt}$ is a \ftrs($\P$).

For a pair $p=(s,t) \in \P$ let $H_p$ denote an optimum (minimum) sized subgraph of $H$ that is a \ftrs($p$), i.e., after any single edge failure $e$, $t$ is reachable from $s$ in $H_p-e$ if and only if that is also the case in $G-e$. 
An optimum sized subgraph of $H$ that is a \ftrs($p$), may not be unique. However we arbitrarily choose one such subgraph, and throughout this section refer to that as $H_p$. The following proposition is immediate from the definition of optimum sized \ftrs.

\begin{proposition}
\label{prop:edge-disj-paths}
For any pair $p=(s,t) \in \P$, $H_p$ is union of two $s-t$ paths intersecting only at $(s,t)$-cut-edges and $(s,t)$-cut-vertices in $H$.
\end{proposition}
It directly follows from the above proposition that, any \emph{minimal} $(s,t)$-cut\footnote{A set of edges is said to be an \emph{minimal} $(s,t)$-cut if and only if it is a valid $(s,t)$-cut and any of its proper subset is not an $(s,t)$-cut.} in $H_p$ is of size at most $2$. The following property of any (simple) $s-t$ path in $H_p$ will be useful in studying the structure of $H$ (especially in proving Claim~\ref{claim:edge-share}).
\begin{observation}
\label{obs:one-cut-edge}
For any pair $p=(s,t) \in \P$ consider a (simple) $s-t$ path $Q$ in $H_p$. Let $C$ be any minimal $(s,t)$-cut in $H_p$. Then $Q$ takes exactly one edge from the set $C$.
\end{observation}

\begin{proof}
The result trivially holds for a minimal $(s,t)$-cut of size one.
Moreover, by minimality of $H_p$, there cannot exists $(s,t)$-cut of size three or larger in $H_p$.
So we are left to consider minimal cuts of size two.

Consider a simple $s-t$ path $Q$ in $H_p$, and let $C=(e,e')$ be a minimal $(s,t)$-cut of size two in $H_p$.
Let $Z_p=(v_0=s,v_1,\ldots,v_\ell=t)$ be the $(s,t)$-cut-vertices in $H_p$, in the order they appear on~$Q$. Let ${\cal I}\subseteq [1,\ell]$ be those indices for which $(v_{i-1},v_i)$ is a cut-edge in $H_p$, and $\cal J$ be $[1,\ell]\setminus {\cal I}$. So for each $i\in {\cal J}$, there exists $2$-edge-disjoint paths from $v_{i-1}$ to $v_i$, in $H_p$; let these be respectively denoted by $R^0_i$ and $R^1_i$. By definition of cut-vertices, it is easy to observe that none of the internal vertices of $R^{0}_i$ and $R^{1}_i$ can lie in $Z_p$, for $i\in {\cal J}$.

Obtain $\widetilde Q$ from $Q$ by replacing $R^{z_i}_i$ with $R^{1-z_i}_i$, for $z_i\in \{0,1\}$ and $i\in {\cal J}$. So $Q$ and $\widetilde Q$, both lie in $H_p$, and intersect only at $(s,t)$-cut-edges and cut-vertices. Now for the minimal $(s,t)$-cut $C=\{e,e'\}$, $Q$ will contain one of the cut-edges, say $e$, and $\widetilde Q$ will contain the other cut edge, i.e. $e'$.
This shows that any simple $s-t$ path in $H_p$ contains exactly one of the edges of a minimal $(s,t)$-cut of size two.
\end{proof}

For each pair $p\in \P$ we define \emph{critical-edge-set} for $p$, denoted by $C_p$, as the set of all edges $e$ in $H$ such that $H- e$ is not a \ftrs($p$). Sometimes we will also use the notation $C_p$ to denote the underlying subgraph formed by edges present in $C_p$. Note, all the edges of $H_{opt}$ must be in $\cup_{p \in \P} C_p$. Otherwise, if there exists an edge $e\in H_{opt}$ that is not in any of $C_p$'s then we can remove $e$ from $H$ while preserving $1$-fault tolerant reachability for all pairs $p \in \P$, leading to a contradiction on the optimality of the size of $H_{opt}$. So we can deduce the following observation.
\begin{observation}
\label{obs:critical}
For any edge $e \in H$, either $e \in C_p$ for some pair $p \in \P$, or $e \in H_{scc}$.
\end{observation}

Let $d_{avg}$ be the average in-degree of $H$, i.e., $d_{avg}=|E(H)|/n$. Recall that we use $E(H)$ to denote the set of edges in the subgraph $H$. Now partition the set of vertices $V$ into two subsets:
\begin{enumerate}
\item The set of \emph{light} vertices $V_{\ell}$ containing all the vertices whose in-degree in $H_{opt}$ is (strictly) less than $d_{avg}/2$; and
\item The set of \emph{heavy} vertices $V_{h}$ containing all the vertices whose in-degree in $H_{opt}$ is at least $d_{avg}/2$.
\end{enumerate}
Let $E_{\ell}$ and $E_h$ respectively be the set of incoming edges to the vertices in $V_{\ell}$ and $V_h$, in graph $H_{opt}$. Clearly, $|E_{\ell}| < (\frac{d_{avg}}{2}) n=|E(H)|/2$.

\begin{lemma}
\label{lem:bound-heavy}
If $d_{avg} > 12$, then for each $p\in \P$, $|E_h \cap C_p| \le 16|\P|/d_{avg}$.
\end{lemma}
 We defer the proof of the above lemma to the end of this section. Now assuming the above we show the desired bound on $|E(H)|$, as stated in Theorem~\ref{thm:existential-UB-1}. First of all, if $d_{avg} \le 12$ then $E(H)$ is of size $O(n)$. So from now on we assume that $d_{avg} > 12$. Note that $|E(H)|=|E(H_{opt})|+|E(H_{scc})|$. By the result of the previous section (see Theorem~\ref{theorem:SCC_1}) we know that  the graph $H_{scc}$ has at most $O(n)$ edges. 
 Observe that,
 \begin{align}
 \label{eqn1}
 &|E(H_{opt})| = |E_{\ell}| + |E_h| \le \frac{|E(H_{opt})|}{2} + |E_h| \nonumber\\
 \Rightarrow &|E(H)| \le 2(|E_h|+|E(H_{scc})|).
 \end{align}
 We bound the size of $E_h$ as follows,
 \begin{align}
 \label{eqn2}
 |E_h| &\le \sum_{p \in \P} |E_h \cap C_p| \qquad \text{(since for all $e \in H_{opt}$, $e\in C_p$ for some $p \in \P$ by Observation~\ref{obs:critical})}\nonumber\\
 & \le \sum_{p\in \P}\frac{16|\P|}{d_{avg}} \qquad \quad \text{(by Lemma~\ref{lem:bound-heavy})} \nonumber\\
 &= \frac{16|\P|^2n}{|E(H)|}.
 \end{align}
 By combining inequalities~(\ref{eqn1}) and~(\ref{eqn2}) we get that $|E(H)| \le O(n+|\P|\sqrt{n}).$

Now it only remains to prove Lemma~\ref{lem:bound-heavy}. Before going into the proof we would first like to make a few important observations regarding the structure of $H$, which will eventually help us to bound its size. Consider a (simple) path $Q$ from $s$ to $t$ where $(s,t)=p \in \P$, such that all the edges on $Q$ are in $H_p$. Since $Q$ is a simple path it gives a natural ordering among the vertices present on it. Let us denote this ordering relationship by $<_Q$ ($\le_Q$ and $>_Q$).

\begin{claim}
\label{claim:edge-share}
Consider a pair $p=(s,t) \in \P$, and let $Q$ be a (simple) $s-t$ path in $H_p$. For an edge $e = (u,v) \in C_p$ on $Q$, suppose there are two vertices $u_1,u_2 <_Q v$, and let $Q_1$ and $Q_2$ be (any arbitrary) $u_1-v$ and $u_2-v$ path respectively, in $H- \{e\}$. Then $Q_1$ and $Q_2$ must share an edge.
\end{claim}
\begin{proof}
Since $e \in C_p$ if we exclude $e$ from $H$, the remaining subgraph $H'=H- \{e\}$ will not be a \ftrs($p$). Observe that $e$ cannot be an $(s,t)$-cut-edge in $H$ as it violates the existence of paths $Q_1$ and $Q_2$ in $H- \{e\}$. So there must exist an edge $f=(u_f,v_f)$ on failure of which there is an $s-t$ path, in $H - \{f\}$, but there is no such path in $H'- \{f\}=H-\{e,f\}$. Since there is no $s-t$ path in $H-\{e,f\}$, $C=\{e,f\}$ must be an $(s,t)$-cut in $H$ (and also in $H_p$). Further, since $e$ is not an $(s,t)$-cut-edge in $H$, $C$ is a minimal $(s,t)$-cut in $H_p$. Let $(A,B)$ be a partition of $V$ induced by the cut $C$ in $H$ such that $s\in A$ and $t\in B$. As $Q$ is a simple $s-t$ path in $H_p$, by Observation~\ref{obs:one-cut-edge} it passes through cut $C$ only once, thereby implying $u_1,u_2\in A$. Thus, the path $Q_i$ from $u_i$ to $v$ must pass through an edge in~$C$, for $i=1,2$. As $Q_1$ and $Q_2$ lie in $H- \{e\}$, they both must pass through the edge $f$, thereby proving the claim.
\end{proof}

Now using the above claim we prove the following.


\begin{claim}
\label{claim:bound-incoming-forward}
For any pair $p=(s,t)\in \P$, let $Q$ be a (simple) $s-t$ path in $H_p$. For a vertex $v$ on $Q$, suppose there are two incoming edges $h_1$ and $h_2$ incident on $v$, which are not part of $Q$. Further, for $i=1,2$, let $p_i=(s_i,t_i)$ be a pair satisfying $h_i\in C_{p_i}$, and let $Q_i\in H_{p_i}$ be an $s_i-t_i$ path containing $h_i$. If $v_i~(\neq v)$ is the last vertex in $Q_i[s_i,v]$ that also lies on $Q$, for $i=1,2$, then we cannot simultaneously have  $v_1<_Q v$ as well as $v_2<_Q v$.
\end{claim}
\begin{proof}
Let us assume on contrary that $v_1,v_2<_Q v$. Without loss of generality, we can assume $v_1 \leq_Q v_2$. 
Let $f$ be an edge on whose failure, any $s_2-t_2$ path must pass through the edge $h_2$. Observe $h_2$ cannot be an $(s_2,t_2)$ cut-edge, as $Q[v_2,v]$ is a $v_2-v$ path avoiding $h_2$, thus $C=\{f,h_2\}$ is a minimal $(s_2,t_2)$ cut in $H$. Note, $f$ must be on $Q[v_2,v]$. 

By Claim~\ref{claim:edge-share}, $Q_1[v_1,v]$ and $Q_2[v_2,v]$ must share an edge. Let $e = (x, y)$ be the last such shared edge. Now even if we exclude $h_2$ from $H$ we still get the following $s_2-t_2$ path: $Q_2[s_2,y] \circ Q_1[y,v] \circ Q_2[v,t_2]$ in $H- C$. Since, $C=\{f,h_2\}$ is an $(s_2,t_2)$ cut, this contradicts the assumption that $v_1,v_2<_Q v$.
\end{proof}

\begin{claim}
\label{claim:bound-incoming-back}
For any pair $p=(s,t)\in \P$, let $Q$ be a (simple) $s-t$ path in $H_p$. For a vertex $v$ on $Q$, suppose there are two incoming edges $h_1$ and $h_2$ incident on $v$, which are not part of $Q$. Consider the sets $B_{h_i}:=\{q \in \P | h_i \in C_{q}\}$, for $i\in\{1,2\}$. Further, for each $i\in\{1,2\}$, suppose for every pair $p'=(s',t') \in B_{h_i}$, the following holds: 
\begin{enumerate}
\item Neither there is an $s'-v$ path in $H_{p'}$ with $h_i$ as its last edge, that is internally vertex-disjoint with $Q$,
\item Nor there exists a vertex, say $v_i <_Q v$, with a $v_i-v$ path in $H_{p'}$ with $h_i$ as its last edge that is edge-disjoint with $Q$.
\end{enumerate}
Then either $h_1 \not \in H_{opt}$ or $h_2 \not \in H_{opt}$. (Recall, $H_{opt}=H- H_{scc}$.)
\end{claim}

\begin{proof}
For $i=1,2$, consider a pair $(s_i,t_i)\in B_{h_i}$. Let $Q_i\in H_i$ be an $s_i-t_i$ path containing $h_i$. Let $u_i\in V- \{v\}$ be the last vertex in $Q_i[s_i,v]$ that also lies on $Q$. Due to the preconditions mentioned in the statement of the claim, such a $u_i$ must exists, and is necessarily contained in segment $Q[v,t]$. Observe $Q_i[u_i,v]$ is internally vertex disjoint with $Q$. 
Next we show that if $u_2 \leq_Q u_1$, $h_2 \not \in H_{opt}$.

Suppose $u_2 \leq_Q u_1$.  Let $f$ be an edge on whose failure, any $s_2-t_2$ path must pass through the edge $h_2$. Observe $h_2$ cannot be an $(s_2,t_2)$ cut-edge, as $Q[u_2,u_1]\circ Q_1[u_1,v]$ is a $u_2-v$ path avoiding $h_2$, thus $C=\{f,h_2\}$ is a minimal $(s_2,t_2)$ cut in $H$. Observation~\ref{obs:one-cut-edge} implies $f\notin Q_2$. Note, $f$ must be either on $Q[u_2,u_1]$ or on $Q_1[u_1,v]$, thereby implying $f\notin Q[v,u_2]$. Thus $u_2$ and $v$ are strongly connected in $H - \{f\}$ as the cycle $Q_2[u_2,v] \circ Q[v,u_2]$ is intact $H - \{f\}$. Hence $H_{scc}$ contains a $u_2-v$ path even after the failure of $f$, and let $R'$ be such a path. Thus even if we exclude $h_2$ from $H_{opt}$ we still get the following $s_2-t_2$ path: $Q_2[s_2,u_2] \circ R \circ Q_2[v,t_2]$ in $H- \{f\}$. So due to optimality of $H_{opt}$, the edge $h_2$ cannot be in $H_{opt}$ ($=H- H_{scc}$).
Similarly when $u_1 <_Q u_2$, $h_1 \not \in H_{opt}$, and this completes the proof.
\end{proof}

Now we are ready to prove Lemma~\ref{lem:bound-heavy}.
\begin{proof}[Proof of Lemma~\ref{lem:bound-heavy}]
For the sake of contradiction let us assume that for some pair $p=(s,t) \in \P$, $|E_h \cap C_p| > \frac{16|\P|}{d_{avg}}$. Recall, by Proposition~\ref{prop:edge-disj-paths} $H_p$ is union of two (simple) $s-t$ paths, say $Q$ and $\widetilde Q$, intersecting only at $(s,t)$-cut-edges in $H$.
At least one of these two paths, say $Q$, must contain at least $\frac{8|\P|}{d_{avg}}$ edges from $|E_h \cap C_p|$. Now consider the following vertex set 
$$Q_h:=\{v \in V_h| \text{ there exists an edge }(u,v)\in Q\text{ that is also in }E_h \cap C_p\}.$$
Clearly, $|Q_h| \ge \frac{8|\P|}{d_{avg}}$. Let $Q_e$ denote the subset of edges in $E(H_{opt})$ that are incident on the vertices in $Q_h$ and not part of the path $Q$. Next consider the following edge-set
\begin{align*}
A:=\{&(u,v)\in Q_e| \text{ for some }(s',t')\in \P \text{ there exists an }s'-v\text{ path in }H \text{ with }(u,v) \text{ as its last}\\ 
&\text{ edge, that is internally vertex-disjoint with }Q\}.
\end{align*}
Observe, it follows from Proposition~\ref{prop:edge-disj-paths} that $|A| \le 2|\P|$. Assuming $d_{avg} > 12$, a simple counting argument shows that there exists a vertex $v\in Q_h$ such that number of edges from $Q_e- A$ that are incident on $v$ is at least $3$. If not, then since each vertex in $Q_h$ is by definition heavy, $|A| \ge |Q_h|(\frac{d_{avg}}{2}-2) \ge \frac{8 |\P|}{3}$ assuming $d_{avg} > 12$, which leads to a contradiction.

It implies that there must exist two vertices $u_1, u_2$ on $Q$ where $u_1 <_Q u_2$, and two edges $h_1, h_2 \in Q_e- A$ (incident on $v$) such that there are $u_1-v$ path, say $Q_1$, with $h_1$ as its last edge and $u_2-v$ path, say $Q_2$, with $h_2$ as its last edge in $H$, where both $Q_1, Q_2$ are edge-disjoint with $Q$. Now either $u_1, u_2 <_Q v$, or $u_1, u_2 > _Q v$.

Since due to optimality of $H_{opt}$, $h_2 \in C_{p'}$ for some $p'=(s',t')\in \P$, without loss of generality we can further assume that the path $Q_2$ is in the subgraph $H_{p'}$. Now if $u_1, u_2 <_Q v$, by Claim~\ref{claim:bound-incoming-forward} there must be an $s'-v$ path in $H$ with $h_2$ as its last edge that is internally vertex-disjoint with $Q[s,v]$. Thus by the definition of set $A$, $h_2 \in A$, leading to a contradiction. Now consider the other alternative, i.e., when $u_1, u_2 > _Q v$. In this case without loss of generality we can assume that for $i=1,2$ there exists no vertex, say $v_i <_Q v$, with a $v_i-v$ path in $H$ with $h_i$ as its last edge, that is edge-disjoint with $Q$. Then by Claim~\ref{claim:bound-incoming-back} either $h_1 \not \in H_{opt}$ or $h_2 \not \in H_{opt}$, which again leads to a contradiction. 

Hence we conclude that for all pairs $p \in \P$, $|E_h \cap C_p| \le \frac{16|\P|}{d_{avg}}$ for  $d_{avg} > 12$.
\end{proof}

\subparagraph{A Constructive Algorithm}
Observe that the size of an minimal subgraph $H$ which is an \ftrs$(\P,G)$, must have size at most $O(n+|\P|\sqrt{n})$. Now, a simple constructive algorithm is as follows: We initialize $H$ to $G$. Next for each pair $(s,t)\in \P$ for each $e\in E(H)$ check if the $(s,t)$-cut-edges in $G$ and $G-\{e\}$ are identical, if so, remove $e$ from $H$. The process terminates in polynomial time and results in a graph which is a minimal \ftrs$(\P,G)$. 

\subsection{Upper Bound II}

In this subsection, we present our second construction for pairwise reachability preservers. 

\begin{theorem}
\label{thm:existential-UB-2}
For any directed graph $G=(V,E)$ with $n$ vertices, $m$ edges, and a set $\P$ of vertex-pairs, there exists a 
poly-time computable
\ftrs($\P,G$) containing at most $O(n\sqrt{|\P|})$ edges. 
\end{theorem}

By Proposition \ref{prop:edge-disj-paths}, for any pair $p=(s,t) \in \P$, $\ftrs(p)$ is union of two $s-t$ paths intersecting only at $(s,t)$ cut-edges and cut-vertices. Let these paths be respectively denoted by $Q_{s,t}$ and $\widetilde Q_{s,t}$.

Let $\C$ be the collection $\C=\bigcup_{(s,t)\in \P}\big\{Q_{s,t}, \widetilde Q_{s,t}\big\}$, and $W$ be initialized to $\emptyset$. For each vertex $v$, let $\freq(v,\C)$ denote the total number of paths in $\C$ that contains $v$. Now we use the following procedure:

\begin{enumerate}
\item While there is a vertex $w$ with $\freq(w,\C)>\sqrt{|\P|}$, we add $w$ to $W$, and remove all those paths from $\C$ that contains $w$.
\item Initialize $H$ to union of subgraphs $\ftrs(w,G)$ and $\ftrs(w,G^R)$, taken over all $w\in W$.
\item Also add to $H$ the union of edges lying in paths remaining in $\C$.
\end{enumerate}
  
The size of set $W$ is at most $2\sqrt{|\P|}$ since each time a vertex is added to $W$ at least $\sqrt{|\P|}$ paths are eliminated from $\C$. By Theorem~\ref{theorem:ftrs} we can bound the size of both $\ftrs(w,G)$ and $\ftrs(w,G^R)$ by $O(n)$ for each $w \in W$. After step 2, since for each $v\in V$, $\freq(v,\C)$ is bounded by $\sqrt{|\P|}$, the number of edges in $H$ is at most $O(n\sqrt{|\P|})$. The correctness follows from the following claim.

\begin{claim}
\label{clm:correct-sub-quadratic}
The subgraph $H$ computed by above process is a \ftrs($\P,G$).
\end{claim}
\begin{proof}
Consider a pair $(s,t)\in \P$ and an edge failure $e\in E$. Observe that if $e$ lies in both $Q_{s,t}$ and $\widetilde Q_{s,t}$, then $e$ must be an $(s,t)$-cut. In such a scenario there will be no path from $s$ to $t$ in $G-\{e\}$ as well as in $H-\{e\}$. So let us assume $e$ does not lie in at least one of the paths, $Q_{s,t}$ or $\widetilde Q_{s,t}$. Without loss of generality, we assume $Q_{s,t}$ will remain intact in $G-\{e\}$. We will show that there is an $s-t$ path in $H-\{e\}$ even when $Q_{s,t}\notin H$. Recall that if $Q_{s,t}\notin H$, then $Q_{s,t}$ must contain a vertex from set $W$, let this vertex be $w$. Since there is a path from $s$ to $t$ in $G-\{e\}$ containing $w$, there must exists an $s-w$ path, say $R_{1}$, in  $\ftrs(V\times \{w\},G)-\{e\}$, and a $w-t$ path, say $R_{2}$, in $\ftrs(\{w\}\times V,G)-\{e\}$. The concatenated path $R_1\circ R_2$ is an $s-t$ path avoiding $e$. Also $R_1\circ R_2$ lies in $H$ as we include in $H$ a $\ftrs(\{x\}\times V,G)$ as well as a $\ftrs(V\times \{x\},G)$, for each $x\in W$. It thus follows that $H$ is a fault-tolerant reachability preserver for the pair $(s,t)$.
\end{proof}




\section{Lower Bounds for Pairwise Reachability Preserver under Dual Failures}
\label{section:lower-bounds}

In this section, we provide several lower bound results. In our constructions, we will employ the following lower bound for pairwise distance preservers that was established by Coppersmith and Elkin in \cite{CE06}, and later reformulated in \cite{AB18} using standard tricks in \cite{Bodwin17}.

\begin{theorem}[\cite{CE06,Bodwin17,AB18}]\label{theorem:lb-distance}
For any integer $d \geq 2$ there are infinitely many $n \in \mathbb{N}$ such that for any $p\in[n^{f(d)},n^{f(d+1)}]$, where $f(d)=\frac{2d^2-2d-2}{(d^2+d-2)}$, there exists an $n$-node undirected unweighted graph $G = (V, E)$ and a pair-set $\P\subseteq V\times V$ of size $|\P|=O(p)$, such that
\begin{itemize}
\item For each pair $(s, t) \in \P$ there is a unique shortest path in $G$ between $s$ and $t$, 
\item These paths are all edge-disjoint and have identical length (which we denote by $L$), and
\item The edge set of $G$ is precisely the union of these paths and has size $\Omega \big(n^{2d/(d^2 +1)} p^{(d^2-d)/(d^2 +1)}\big)$.
\end{itemize}
\end{theorem}

\noindent	
Throughout this section, we will reserve $f(d)$ to refer to the function $\frac{2(d^2-d-1)}{(d^2+d-2)}$.






The construction of our lower-bound graph $G=(V_G,E_G)$ is as follows. Let $H=(V_H,E_H),~\P_H$ be an instance drawn from Theorem~\ref{theorem:lb-distance}. Let $L$ be the common distance between all the pairs in $\P_H$, and $K=L^r$ for some parameter $r\geq 1$ to be fixed later on. For each node $u$ in $H$, add $2K$ copies of $u$ to $G$, namely, $u_1,\ldots,u_{2K}$. For each edge $(u,v)$ in $H$ and $i\geq 2$, add edges $(u_{i},v_{i-1})$ and $(v_{i},u_{i-1})$ to $G$. 

Next, for each node $v\in V_H$, add two paths $v_{\textsc{left}}:=(v_{\textsc{left},1},\ldots, v_{\textsc{left},2K})$ and $v_{\textsc{right}}:=(v_{\textsc{right},1},\ldots, v_{\textsc{right},2K})$, each on a set of $2K$ new nodes to $G$. Also, add an edge from $v_{\textsc{left},i}$ to $v_{i}$, and $v_i$ to $v_{\textsc{right},i}$ to $G$, for $i\in[1,2K]$.

Finally, for each $(x,y)\in \P_H$, create two new vertices $s_{x,y}$ and $t_{x,y}$, and include $(s_{x,y}, t_{x,y})$ in pair-set $\P_G$; add directed edges from $s_{x,y}$ to $x_{\textsc{left},1},y_{\textsc{right},1}$ and add directed edges from $x_{\textsc{left},2K},y_{\textsc{right},2K}$ to $t_{x,y}$. 
This completes the description of $G$, and pair-set $\P_G$.

\begin{figure}[!ht]
\centering
\includegraphics[scale=.5]{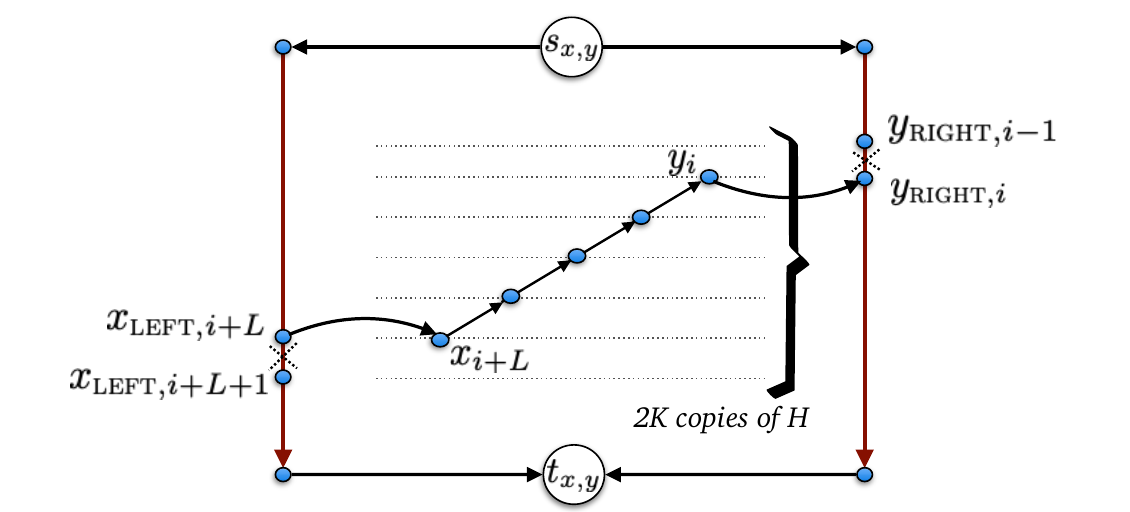}
\caption{Dual fault-tolerant reachability preserver: depiction of graph $G$ and pair $(s_{x,y},t_{x,y})\in \P_G$.} 
\end{figure}

%
%
%
%
%
%
Observe that 
\begin{enumerate}
\item $|\P_G|=|\P_H|$;
\item $|V_G|=\Theta(|\P_H| + L^r |V_H|)=\Theta(L^r|V_H|)$, whenever $|\P_H|\leq O(L^r |V_H|)$; and
\item $|E_G|=\Theta(|\P_H| +L^r |E_H|)=\Theta(L^{r+1}|\P_H|)$ (since $|E_H|=L|\P_H|$ by the description of $H$ given in Theorem~\ref{theorem:lb-distance}).
\end{enumerate}

We first analyze the size of $E_G$. By Theorem~\ref{theorem:lb-distance}, we have
\begin{equation}\label{eq:L}
L=\frac{|E_H|}{|\P_H|}=\Theta(|V_H|^\frac{2d}{d^2+1}~|\P_H|^\frac{-(d+1)}{d^2+1}).
\end{equation}

Let $n:=|V_G|=\Theta(|V_H|\cdot L^r)$ and $m:=|E_G|=\Theta(|\P_H|\cdot L^{r+1})$. On multiplying $L^\frac{2dr}{d^2+1}$ on both sides of Equation~\ref{eq:L}, we obtain

\begin{align}
L^{1+\frac{2dr}{d^2+1}}&=
	\Theta\Big(|V_G|^\frac{2d}{d^2+1}~|\P_G|^\frac{-(d+1)}{d^2+1}\Big) 
	\nonumber\\
\text{or,~}L^{r+1}&=
	\Theta\Big(|V_G|^\frac{2d(r+1)}{d^2+2rd+1}~|\P_G|^\frac{-(d+1)(r+1)}{d^2+2rd+1}\Big)
	\nonumber \\
\text{Thus,~}|E_G|=\Theta(|\P_G|& L^{r+1})
	=\Theta\Big(|V_G|^\frac{2d(r+1)}{d^2+2rd+1}~|\P_G|^\frac{d^2+dr-(d+r)}{d^2+2rd+1}\Big)
	\label{eq:bound-2ft}
\end{align}

\noindent
The size of $V_H$ is given by $|V_H|=\Theta\Big(\frac{|V_G|}{L^r}\Big)=\Theta\Big(n^\frac{d^2+1}{d^2+2rd+1}|\P_G|^\frac{(d+1)r}{d^2+2rd+1}\Big).$

\noindent
So, the condition $|V_H|^{f(d)}\leq  |\P_H|\leq |V_H|^{f(d+1)}$ translates to 

$$
\Big(n^\frac{d^2+1}{d^2+2rd+1}|\P_G|^\frac{(d+1)r}{d^2+2rd+1}\Big)^{f(d)}\leq |\P_G|\leq
\Big(n^\frac{d^2+1}{d^2+2rd+1}|\P_G|^\frac{(d+1)r}{d^2+2rd+1}\Big)^{f(d+1)}
$$

which can be re-stated as: $n^\frac{(d^2+1)f(d)}{d^{2}+2rd+1-(d+1)rf(d)}\leq |\P_G|\leq n^\frac{(d^2+1)f(d+1)}{d^{2}+2rd+1-(d+1)rf(d+1)}$

and on simplification is equivalent to 

\begin{equation}\label{eq:new:f(d)}
n^\frac{2(d^2-d-1)}{(d^2+d+2r-2)}\leq |\P_G|\leq n^\frac{2(d^2+d-1)}{(d^2+3d+2r)}.
\end{equation}

We now prove that a dual fault-tolerant reachability preserver of $G$ requires $\Omega(|E_G|)$ edges. Consider pair $(s_{x,y},t_{x,y})$ in set $\P_G$, for a pair $(x,y)\in \P_H$. Let $\pi(x,y,H)=(x=w^0,w^1,\ldots,w^L=y)$ be the shortest path between $x$ and $y$ in the undirected graph $H$. 

By construction of $G$ and uniqueness of $\pi(x,y,H)$, it follows that for any $i\in[1,K]$, there exists a unique path from $x_{i+L}$ to $y_i$ in $G$. Indeed $\pi(x_{i+L},y_i,G	)=(w^0_{i+L},w^1_{i+L-1},$ ~$w^2_{i+L-2},\ldots,w^L_{i})$,  is the shortest and the only path starting from $x_{i+L}=w^0_{i+L}$ and terminating to $y_i=w^L_{i}$ in $G$. Moreover, there is no path from $x_{i-\alpha+L}$ that terminates to $y_{i+\beta}$ for non-negative integers $\alpha,\beta$ if at least one of them is greater than $0$.

On failures of edges $(x_{\textsc{left},i+L},x_{\textsc{left},i+L+1})$ and $(y_{\textsc{right},i-1},y_{\textsc{right},i})$, the concatenated path 
$$(s_{x,y}, x_{\textsc{left},1},\ldots,x_{\textsc{left},i+L})\circ  \pi(x_{i+L},y_i,G)\circ (y_{\textsc{right},i},\ldots,y_{\textsc{right},2L},t_{x,y})$$
is the only path from $s_{x,y}$ to $t_{x,y}$ in the surviving part of $G$. 

Thus, a dual fault-tolerant reachability preserver of $G$ must contain $\pi(x_{i+L},y_i,G)$, for each $i\in[1,K]$ and $(x,y)\in \P_H$. From the fact that the shortest path between pairs in $\P_H$ are all edge-disjoint in $H$, it directly follows, a dual fault-tolerant reachability preserver of $G$ must contain $\Omega(KL|\P_H|)=\Omega(K|E_H|)=\Omega(|E_G|)$ edges. The above analysis along with Eq.~\ref{eq:bound-2ft} and Eq.~\ref{eq:new:f(d)} proves our main result, Theorem~\ref{theorem:lb-dual-generic}, 
on dual failure.

\begin{theorem}\label{theorem:lb-dual-generic}
For any integer $d \geq 2$, any real $r\geq 1$, any real $c\in (0,1)$, there are infinitely many $n \in \mathbb{N}$ such that for any $p\in[n^{f(d,r)},\min\{cn,n^{f(d+1,r)}\}]$, where $f(d,r)=\frac{2d^2-2d-2}{(d^2+d+2r-2)}$, there exists an $n$-node directed graph $G = (V, E)$ and pair-set $\P\subseteq V\times V$ of size $|\P|=O(p)$, such that any dual fault-tolerant reachability preserver of $G$ for $\P$,
must have $\Omega \Big(n^\frac{2d(r+1)}{d^2+2rd+1}~|\P|^\frac{(d-1)(d+r)}{d^2+2rd+1}\Big)$ edges.
\end{theorem}
Note, in the above theorem we need $p \le cn$ for $c<1$, to ensure that $|\P_H|\leq O(L^r |V_H|)$ (See the construction of $G$ from $H$). Some instances of the above theorem are as below.
\begin{itemize}
\item $\Omega(n^{{8}/{9}}|\P|^{1/3})$ edges for $n^{1/3}\leq |\P|\leq n^{5/6}$  (when $d=2, r=1$).
\item $\Omega(n^{12/13}|\P|^{4/13})$ edges for $n^{1/4}\leq |\P|\leq n^{5/7}$  (when $d=2, r=2$).
\end{itemize}

\subparagraph*{Non-existence of Linear-sized Dual Fault-Tolerant Preservers}
For $d=2$, the lower bound turns to be $\Omega \Big(n^\frac{4r+4}{4r+5}~|\P|^\frac{r+2}{4r+5}\Big)$ on the size of preserver, and bound on $\P$ becomes $[n^\frac{1}{r+2},n^\frac{5}{r+5}]$. Let $\epsilon=\frac{2}{(r+2)}$, so $\epsilon\leqslant 2/3$ for $r\geqslant 1$. Now for $|\P|=n^\epsilon$, observe $|\P|=n^\frac{2}{r+2}$ which lies in range $[n^\frac{1}{r+2},n^\frac{5}{r+5}]$, the lower bound on the size of preserver becomes $\Omega(n^{1+\frac{1}{4r+5}})$
which is $\Omega(n^{1+{\epsilon}/{8}})$. Thus the following non-linearity result is immediate.

\begin{theorem}
For every $p=n^\epsilon$, for $\epsilon\leq 2/3$, there is an infinite family of $n$-node directed graphs and pair-sets $\P$ of size $|\P|=p$, for which every dual fault-tolerant reachability preserver of $G$ for $\P$ requires at least $\Omega(n |\P|^{\frac{1}{8}})$ edges.
\end{theorem}
Recall, Theorem~\ref{theorem:ftrs} implies that for any $\P$ of size $p$ there exists a dual fault-tolerant reachability preserver with at most $O(np)$ edges.
Our result proves a wide separation in the size pairwise 1-fault-tolerant and 2-fault-tolerant reachability preservers.


\section{Application of FT-SCC Preserver in Connectivity Certificates}
\label{section:applications}

In this section we present an  application of $k$-FT-SCC preserver for vertex and edge connectivity 
certificates for digraphs.


\begin{definition}[$k$-connectivity certificate]
For a graph $G=(V,E)$, a subgraph $H=(V,E_H)$ of $G$ is said to be a $k$-Edge (Vertex) Connectivity Certificate of $G$ 
if for each pair of vertices $x,y\in V$, if there are are at least $k$-edge (vertex) disjoint paths from $x$ to $y$, and vice versa in $G$, then the same holds true for graph $H$ as well. 
\end{definition}

Georgiadis et al. \cite{GeorgiadisILP15-soda, GeorgiadisILP15-icalp} showed that for any strongly-connected graph
we can compute a sparse certificate for $2$-vertex connectivity and 2-edge-connectivity comprising of just $O(n)$ edges.
However, little is known about extremal size bound of $k$-connectivity certificates in {\em digraphs}, for a general $k$.

We below present a generic reduction from $(k-1)$-FT-SCC preserver to $k$-connectivity in digraphs.

\begin{lemma}
\label{lemma:FTscc-connetivity-reduction}
Let $H$ be a $(k-1)$-FT-SCC preserver of a digraph $G$, for some integer $k\geq 1$.
Then, $H$ is a $k$-Edge (Vertex) Connectivity Certificate for $G$.\footnote{
Note that reverse is not always true, i.e. a $k$-edge(vertex) connectivity certificate is not a $(k-1)$-FT-SCC preserver.}
\end{lemma}
\begin{proof}
Let $H$ be a $(k-1)$-FT-SCC preserver of $G$. Consider a pair of vertices $x,y\in V$, that are at least $k$-edge (vertex) connected in $G$. 

By Menger's theorem, it follows that for $x$ and $y$ to be $k$-edge (vertex) connected in $G$ it holds that on removal of any set $F\subseteq E$ of $k-1$ edges (resp. $F\subseteq V$ of $k-1$ vertices) from $G$, the surviving graph $G-F$ still contains a path from $x$ to $y$, and a path from $y$ to $x$, i.e. $x$ and $y$ are strongly connected in $G-F$.
Now by definition of $(k-1)$-FT-SCC preserver, we have that on removal of any set $F$ of $k-1$ edges or vertices from $H$, 
 $x$ and $y$ must be strongly connected in $H-F$. So another application of Menger's theorem, proves that $x$ and $y$ to be $k$-edge (vertex) connected in $H$. 
\end{proof}

Thus, Lemma~\ref{lemma:FTscc-connetivity-reduction} together with the FT-SCC preserver construction
presented in Theorem~\ref{theorem:kFT-scc-preserver} provides us a $k$-connectivity certificate of size sub-quadratic in $n$, for any $k\geq 3$, as follows.

\begin{theorem}
There is a polynomial time (randomized) algorithm that given any directed graph $G = (V, E)$ on $n$ vertices and $k\geq 2$, computes a $k$-Edge (Vertex) Connectivity Certificate of $G$ containing at most $\tO(k~2^k~n^{2-\frac{1}{k-1}})$ edges with probability at least $1-1/n^4$.
\end{theorem}

\section{Conclusion}
In this paper we discuss the problem of sparsifying a directed graph while preserving its strong-connectivity and pairwise reachability structure under edge failures. For SCC preservers, we provide a construction of a truly sub-quadratic (in number of vertices) sized subgraph that preserves SCC components under constantly many edge failures. More specifically, we show an upper bound of $\tilde{O}(k 2^k n^{2-1/k})$ on the size for any $n$-node graph with $k$ faulty edges, whereas we show a lower bound of $\Omega(2^k n)$. We would like to pose the problem of closing this gap between upper and lower bound as an open problem.

In case of reachability preserver, we show an upper bound of $O\big(n+\min(|\P|\sqrt{n},~n\sqrt {|\P|})~\big)$ for any $n$-node graph and a vertex-pair set $\P$ with one faulty edge. This implies linear sized preserver for $O(\sqrt{n})$ many vertex-pairs, which is also known to be the limit for standard non-fault tolerant static setting. Unfortunately we do not know how to generalize our single fault-tolerant pairwise reachability preserver construction to dual fault-tolerant setting. On the contrary, we show a striking difference between single and dual fault tolerant setting by proving a super linear lower bound on dual fault-tolerant reachability preserver for $\Omega(n^{\epsilon})$ (for some $\epsilon > 0$) vertex-pairs. One extremely interesting open problem is to get any non-trivial (better than $O(n |\P|)$) upper bound on size of a multiple fault-tolerant pairwise reachability preserver. Other future work lies in improving our size bounds, extending our result to bi-connectivity and other pairwise structures.

\paragraph*{Acknowledgements.} Authors would like to thank Robert Krauthgamer and Greg Bodwin for some useful discussion.

%
%

\bibliographystyle{plainurl}
\bibliography{references}

\appendix
\section{Lower bound of $\Omega(2^k n)$ on the size of k-FT-SCC preserver}
\label{section:lower_bound}


We will now show that a lower bound of  $\Omega(2^k n)$ on the size of $\kftrs$
with respect to designated source, can be directly extended to obtain a similar lower bound
for $k$-FT-SCC preserver. The construction of graph $G$ for a given choice of $k$ is as follows.
Take a perfect binary tree $T$ of height $k$ rooted at $s$. Let $X$ be the set of leaf nodes of $T$, thus $|X|=2^k$. Take another set $Y$ of $n$ (new) vertices. To obtain $G$, add an edge from each $x\in X$ to each $y\in Y$, and an edge from each $y\in Y$ to vertex $s$. In other words, $V(G)=V(T)\cup Y$ and $E(V)=E(T)\cup (X\times Y)\cup (Y\times \{s\})$.


In order to prove that a $k$-FT-SCC preserver for $G$ must contain all edges of $G$,
consider an edge $(x,y)\in X\times Y$.
Let $P$ be the tree path from $s$ to leaf node $x$ of $T$. Let $F$ be the set of all those edges $(u,v)\in T$
such that $u\in P$ and $v$ is the $child$ of $u$ not lying on $P$. Clearly $|F|=k$. Observe that $x$ is
the only leaf node of $T$ reachable from $s$ on the failure of the edges in set $F$.
Thus $P_1=P\circ (x,y)$ is the unique path from $s$ to $y$ in $G- F$, and $P_2=(y,s)$
is the unique path from $y$ to $s$ in $G-F$. Both $P_1$ and $P_2$ must be contained in 
a FT-SCC preserver since $s$ and $y$ are strongly connected in $G-F$.
This shows that all the edges of $G$
must be present in a $k$-FT-SCC preserver, thereby, establishing a lower bound of $\Omega(2^k n)$.

\begin{theorem}
For any positive integers $n,k$ with $n\ge 2^k$, there exists a directed graph on $n$ vertices whose $k$-FT-SCC preserver must have $\Omega(2^kn)$ edges.
\end{theorem}

\end{document}